\documentclass{article}
\usepackage{graphicx}
\usepackage{amssymb,amsmath,amsthm}
\usepackage{charter,eulervm}
\usepackage{authblk}
\usepackage{ascmac}
\newcommand{\bcs}{{\rm bcs}}
\newcommand{\wbcs}{{\rm wbcs}}
\newcommand{\tw}{{\rm tw}}
\newcommand{\opt}{{\rm opt}}
\newcommand{\Left}{{\rm left}}
\newcommand{\Pref}{{\rm pref}}

\newtheorem{lemma}{Lemma}
\newtheorem{theorem}{Theorem}

\begin{document}
\title{Algorithms and Hardness Results for the Maximum Balanced Connected Subgraph Problem}
%
%
\author{Yasuaki~Kobayashi, Kensuke~Kojima, Norihide~Matsubara, Taiga~Sone, Akihiro~Yamamoto}

\maketitle              
\begin{abstract}
The Balanced Connected Subgraph problem (BCS) was recently introduced by Bhore et al. (CALDAM 2019).
In this problem, we are given a graph $G$ whose vertices are colored by red or blue.
The goal is to find a maximum connected subgraph of $G$ having the same number of blue vertices and red vertices.
They showed that this problem is NP-hard even on planar graphs, bipartite graphs, and chordal graphs.
They also gave some positive results: BCS can be solved in $O(n^3)$ time for trees and $O(n + m)$ time for split graphs and properly colored bipartite graphs, where $n$ is the number of vertices and $m$ is the number of edges.

In this paper, we show that BCS can be solved in $O(n^2)$ time for trees and $O(n^3)$ time for interval graphs.
The former result can be extended to bounded treewidth graphs.
We also consider a weighted version of BCS (WBCS). We prove that this variant is weakly NP-hard even on star graphs and strongly NP-hard even on split graphs and properly colored bipartite graphs, whereas the unweighted counterpart is tractable on those graph classes.
Finally, we consider an exact exponential-time algorithm for general graphs.
We show that BCS can be solved in $2^{n/2}n^{O(1)}$ time. This algorithm is based on a variant of Dreyfus-Wagner algorithm for the Steiner tree problem.

\end{abstract}
\section{Introduction}\label{sec:intro}

{\em Fairness} is one of the most important concepts in recent machine learning studies and numerous researches concerning ``fair solutions'' have been done so far such as fair bandit problem~\cite{fbandit}, fair clustering~\cite{fcluster}, fair ranking~\cite{franking}, and fair regression \cite{fregression}.
This brings us to a new question: Is it easy to find ``fair solutions'' in classical combinatorial optimization problems?
Chierichetti et al.~\cite{fairmatching} recently addressed a fair version of matroid constrained optimization problems and discussed polynomial time solvability, approximability, and hardness results for those problems.

In this paper, we study the problem of finding a ``fair subgraph''.
Here, our goal is to find a maximum cardinality connected ``fair subgraph'' of a given bicolored graph.
To be precise, we are given a graph $G = (B \cup R, E)$ in which the vertices in $B$ are colored by blue and those in $R$ are colored by red.
We say that a subgraph is {\em balanced} if it contains an equal number of blue and red vertices.
The objective of the problem is to find a balanced connected subgraph with the maximum number of vertices.
This problem is called the balanced connected subgraph problem (BCS), recently introduced by Bhore et al. \cite{bcs}.
Although finding a maximum size connected subgraph is trivially solvable in linear time, 
they proved that BCS is NP-hard even on bipartite graphs, on chordal graphs, and on planar graphs.
They also gave some positive results on some graph classes: BCS is solvable in polynomial time on trees, on split graphs, and on properly colored bipartite graphs.
In particular, they gave an $O(n^3)$ time algorithm for trees, where $n$ is the number of vertices of the input tree.

BCS can be seen as a special case of the graph motif problem in the following sense.
We are given a vertex (multi)colored graph $G = (V, E, c)$ with coloring function $c: V \rightarrow \{1, 2, \ldots, q\}$ and a multiset $M$ of colors $\{1, 2, \ldots, q\}$.
The objective of the graph motif problem is to find a connected subgraph $H$ of $G$ that agrees with $M$: the multiset $c(H) = \{c(v) : v \in V(H)\}$ coincides with $M$.
If $M$ is given as a set of $k/2$ red colors and $k/2$ blue colors, a feasible solution of the graph motif problem is a balanced connected subgraph of $k$ vertices.
Bj\"orklund~et~al.~\cite{motif} proved that there is an $O^*(2^{|M|})$ time randomized algorithm for the graph motif problem, where the $O^*$ notation suppresses the polynomial factor in $n$.
This allows us to find a balanced connected subgraph of $k$ vertices in time $O^*(2^{k})$ and hence BCS can be solved in $\max\{O^*(2^{0.773n}), O^*(2^{H(1-0.773)n})\} \subseteq O(1.709^n)$ time by using this $O^*(2^{k})$ time algorithm for $k \le 0.773n$ or by guessing the complement of an optimal solution for otherwise, where $H(x) = -x\log_2x - (1-x)\log_2(1-x)$ is the binary entropy function.

In this paper, we improve the previous running time $O(n^3)$ to $O(n^2)$ for trees and also give a polynomial time algorithm for interval graphs, which is in sharp contrast with the hardness result for chordal graphs.
The algorithm for trees can be extended to bounded treewidth graphs.
These results are given in Section~\ref{sec:tree} and \ref{sec:interval}.
For general graphs, we show in Section~\ref{sec:general} that BCS can be solved in $O^*(2^{n/2}) = O(1.415^n)$ time.
The idea of this exponential-time algorithm is to exploit the Dreyfus-Wagner algorithm \cite{DW} for the Steiner tree problem.
Let $R$ be the set of red vertices of $G$. Then, for each $S \subseteq R$ and for each $v$ in $G$, we first compute a tree $T$ that contains all the vertices $S \cup \{v\}$ but excludes all the vertices $R \setminus (S \cup \{v\})$.
This can be done in $O^*(2^{|R|})$ time by the Dreyfus-Wagner algorithm and its improvement due to Bj\"orklund et al.~\cite{subsetconv}.
Once we have such a tree for each $S$ and $v$, we can in linear time compute a balanced connected subgraph $H$ with $V(H) \cap R = S$.
We also consider a weighted counterpart of BCS, namely WBCS: the input is vertex-weighted bicolored graph and the goal is to find a maximum weight connected subgraph $H$ in which the total weights of red vertices and of blue vertices are equal.
If every vertex has a unit weight, the problem exactly corresponds to the normal BCS, and hence the hardness results for BCS on bipartite, chordal, and planar graphs also hold for WBCS. 
In contrast to the unweighted case, WBCS is particularly hard.
In Section~\ref{sec:hardness}, we show that WBCS is (weakly) NP-hard even on properly colored star graphs and strongly NP-hard even on split and properly colored bipartite graphs.
The hardness result for stars is best possible in the sense that WBCS on trees can be solved in pseudo-polynomial time.

\section{Preliminaries}
Throughout the paper, all the graphs are simple and undirected.
Let $G$ be a graph. We denote by $V(G)$ and $E(G)$ the set of vertices and edges in $G$, respectively.
We use $n$ to denote the number of vertices of the input graph. 
We say that a vertex set $U \subseteq V(G)$ is {\em connected} if its induced subgraph $G[U]$ is connected.
A graph is {\em bicolored} if every vertex is colored by blue or red.
Note that this coloring is not necessarily proper, that is, there may be adjacent vertices having the same color.
We denote by $B$ (resp. $R$) the set of blue (resp. red) vertices of the input graph.
The problems we consider are as follows.


\bigskip
\begin{itembox}[l]{\bf{The Balanced Connected Subgraph Problem (BCS)}}
\begin{description}
    \item[Input:] A bicolored graph $G = (B \cup R, E)$.
    \item[Output:] A maximum size connected induced subgraph $H$ of $G$ such that $|V(H) \cap B| = |V(H) \cap R|$.
\end{description}
\end{itembox}
\bigskip
\begin{itembox}[l]{\bf{The Weighted Balanced Connected Subgraph Problem (WBCS)}}
\begin{description}
    \item[Input:] A bicolored vertex weighted graph $G = (B \cup R, E, w)$, $w : B \cup R \rightarrow \mathbb{N}$.
    \item[Output:] A maximum weight connected induced subgraph $H$ of $G$ 
    such that 
    \[
    \sum_{v \in V(H) \cap B} w(v) = \sum_{v \in V(H) \cap R} w(v).
    \]
\end{description}
\end{itembox}
Here, the size and the weight of a subgraph are measured by the number of vertices and the sum of the weight of vertices in the subgraph, respectively. 

\section{Trees and bounded treewidth graphs}\label{sec:tree}

This section is devoted to showing that BCS can be solved in $O(n^2)$ time for trees, which improves upon the previous running time $O(n^3)$ of \cite{bcs}.
We also give an algorithm for BCS on bounded treewidth graphs whose running time is $O(n^2)$ as well.

The essential idea behind our algorithm is the same as one in \cite{bcs}.
Let $T$ be a bicolored rooted tree.
For each $v \in V(T)$, we denote by $T_v$ the subtree of $T$ rooted at $v$.
For the sake of simplicity, we convert the input tree $T$ into a rooted binary tree by adding uncolored vertices as follows.
For each $v \in V(T)$ having $p > 2$ children $v_1, v_2, \ldots, v_p$, we introduce a path of $p - 2$ vertices $u_1, u_2, \ldots, u_{p-2}$ that are all uncolored and make $T$ a rooted binary tree $T'$ as in Fig.~\ref{fig:binarization}.

\begin{figure}
    \centering
    \includegraphics[width=0.55\textwidth]{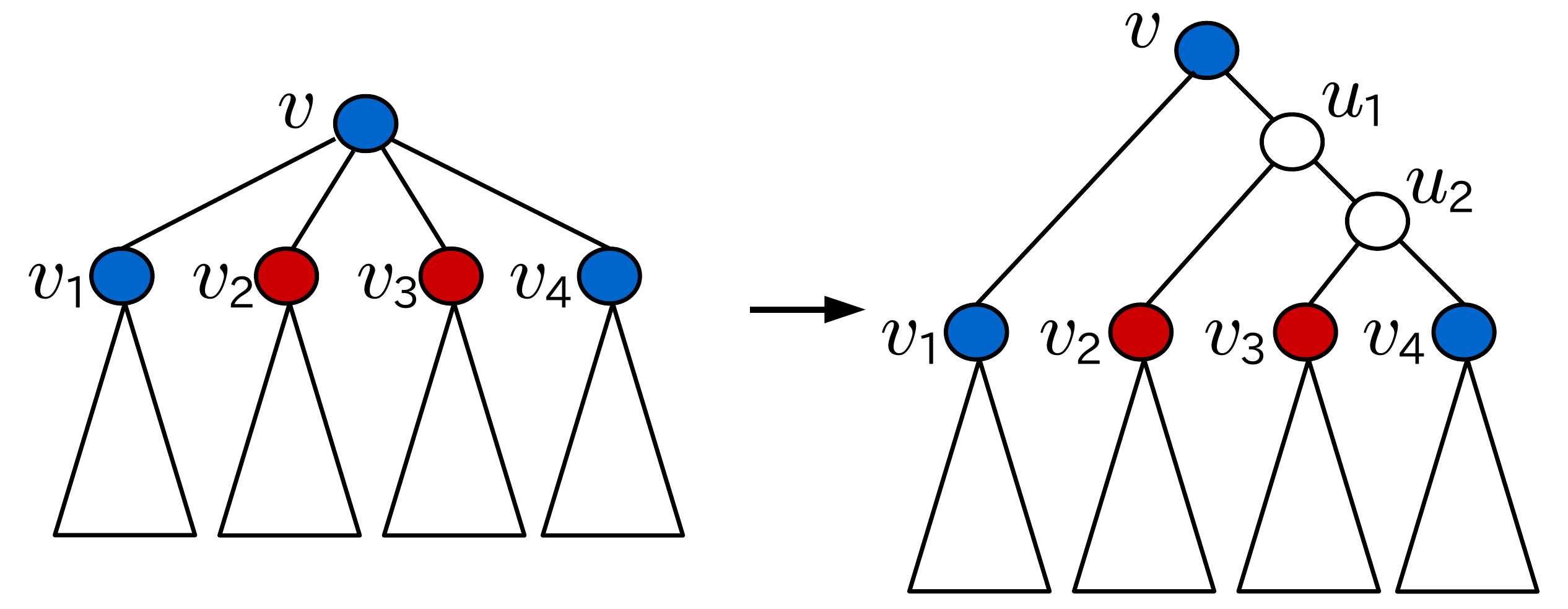}
    \caption{Binarizing a degree-$p$ vertex with $p > 2$.}
    \label{fig:binarization}
\end{figure}
We also assume that each internal node has exactly two children by appropriately adding uncolored children.
This conversion can be done in $O(n)$ time.
It is not hard to see that $T$ has a balanced connected subtree of size $2k$ whose root is $v \in V(T)$ if and only if $T'$ has a connected subgraph with $k$ blue vertices and $k$ red vertices that contains $v$ as the root.
Moreover, $T'$ has $O(n)$ vertices.
Thus, in the following, we seek a connected subgraph with $k$ blue vertices and $k$ red vertices, where $k$ is as large as possible. 
For each $v \in V(T'_v)$ and each integer $-|V(T'_v)| \le d \le |V(T'_v)|$, we say that $S \subseteq V(T')$ is {\em feasible for $(v, d)$} if it satisfies
\begin{itemize}
    \item if $S$ is not empty, $S$ must contain $v$ and
    \item $|S \cap B| - |S \cap R| = d$.
\end{itemize}
We denote by $\bcs(v, d)$ the maximum size of $S$ that is feasible for $(v, d)$, where the size is measured by the number of colored vertices only. 
Let us note that for any $v \in V(T')$, the empty set is feasible for $(v, d)$ when $d = 0$. 
Given this, our goal is to compute $\bcs(v, d)$ for all $v$ and $d$.

Let $f: V(T') \rightarrow \{-1, 0, 1\}$ be a function, where $f(v) = 1$ if $v$ is a blue vertex, $f(v) = -1$ if $v$ is a red vertex, and $f(v) = 0$ if $v$ is an uncolored vertex.

Suppose $v$ is a leaf of $T'$.
Then, $\bcs(v, f(v)) = |f(v)|$, $\bcs(v, 0) = 0$, and $\bcs(v, d) = -\infty$ for $d \notin \{0, f(v)\}$.
Suppose otherwise that $v$ is an internal node.
Let $v_l$ and $v_r$ be the children of $v$.
Observe that every feasible solution for $(v, d)$ can be split by $v$ into two feasible solutions for $(v_l, d_l)$ and $(v_r, d_r)$ for some $d_l$ and $d_r$.
Conversely, for every pair of feasible solutions for $(v_l, d_l)$ and $(v_r, d_r)$, we can construct a feasible solution for $(v, d_l + d_r + f(v))$.
Thus, we have the following straightforward lemma.

\begin{lemma}
    Let $d$ be an integer with $-|V(T'_v)| \le d \le |V(T'_v)|$. Then,
    \[
        \bcs(v, d) = \max_{\substack{d_l, d_r\\d_l + d_r + f(v) = d}} \left(\bcs(v_l, d_l) + \bcs(v_r, d_r)\right) + |f(v)|,
    \]
    where the maximum is taken among all pairs $d_l, d_r$ with $d_l + d_r + f(v) = d$.
\end{lemma}

The running time of evaluating this recurrence may be estimated to be $O(n^3)$ in total to compute $\bcs(v, d)$ for all $v \in V(T')$ and $-|V(T'_v)| \le d \le |V(T'_v)|$ since there are $O(n^2)$ subproblems and solving each subproblem may take $O(n)$.
However, for a node $v$ with two children $v_l$ and $v_r$ the evaluation can be done in $O(|V(T'_{v_l})||V(T'_{v_r})|)$ time in total for all $d$ by simply joining all the pairs $\bcs(v_l, d_l)$ and $\bcs(v_r, d_r)$.
Therefore, the total running time is
\[
    \sum_{v \in V(T')} O(|V(T'_{v_l})||V(T'_{v_r})|) = O(n^2).
\]
This upper bound follows from the fact that the left-hand side can be seen as counting the number of edges in the complete graph on $V(T'_v)$.

\begin{theorem}
    BCS on trees can be solved in $O(n^2)$ time.
\end{theorem}

This algorithm can be extended for bounded treewidth graphs.
{\em Treewidth} is a well-known graph invariant, measuring ``tree-likeness'' of graphs.
A {\em tree decomposition} of a graph $G = (V, E)$ is a rooted tree $T$ that satisfies the following properties:
(1) for each $v \in V(T)$, some vertex set $X_v$, called a {\em bag}, is assigned and $\bigcup_{v \in V(T)}X_v = V$;
(2) for each $e \in E$, there is $v \in V(T)$ such that $e \subseteq X_v$;
(3) for each $w \in V$, the set of nodes $v \in V(T)$ containing $w$ (i.e. $w \in X_v$) induces a subtree of $T$.
The {\em width} of $T$ is the maximum size of its bag minus one.
The {\em treewidth} of $G$, denoted by $\tw(G)$, is the minimum integer $k$ such that $G$ has a tree decomposition of width $k$.

The algorithm is quite similar to dynamic programming algorithms based on tree decompositions for connectivity problems \cite{param}, such as the Steiner tree problem and the Hamiltonian cycle problem.
It is worth noting that the property of being balanced may not be able to be expressed by a formula in the Monadic Second Order Logic (MSO) with bounded length.
Thus, we cannot directly apply the famous Courcelle's theorem \cite{mso1,mso2} to our problem.

Here, we only sketch an overview of the algorithm and the proof is almost the same as those for connectivity problems.
Let $T$ be a tree decomposition of $G$ whose width is $O(\tw(G))$.
Such a decomposition can be obtained in $2^{O(\tw(G))}n$ time by the algorithm of Bodlaender et al. \cite{tddp}.
We can assume that $T$ is rooted.
For each bag $X$ of $T$, we denote by $T_{X}$ the subtree rooted at $X$ and by $V_X$ the set of vertices appeared in some bag of $T_{X}$. 
For each bag $X$ of $T$, integer $d$ with $-|V_X| \le d \le |V_X|$, $S \subseteq X$, and a partition $\mathcal S$ of $S$,
we compute the maximum size $\bcs(X, d, S, \mathcal S)$ of a set of vertices $U \subseteq V_X$ such that $|U \cap B| - |U \cap R| = d$, $X \cap U = S$, and $u, v \in S$ is connected in $U$ if and only if $u$ and $v$ belong to the same block in the partition $\mathcal S$.
We can compute $\bcs(X, d, S, \mathcal S)$ guided by the tree decomposition in $2^{O(\tw(G)\log \tw(G))}n^2$ time.
To improve the running time, we can apply the rank-based approach of Bodlaender et al.~\cite{tddp} to this dynamic programming in the same way as the Steiner tree problem.
The running time is still quadratic in $n$ but the exponential part can be improved to $2^{O(\tw(G))}$.

\begin{theorem}
    BCS can be solved in $2^{O(\tw(G))}n^2$ time.
\end{theorem}

We can extend the algorithm for trees to the weighted case, namely WBCS.
For a tree $T$, instead of computing $\bcs(v, d)$, we compute $\wbcs(v, d)$; For $v \in V(T)$ and for $d$ with $-\sum_{u \in V(T_v)}w(u) \le d \le \sum_{u \in V(T_v)}w(u)$, $\wbcs(v, d)$ is the maximum total weight of a subtree $T_v'$ of $T_v$ that contains $v$ and satisfies $\sum_{u \in V(T_v')\cap B}w(u) - \sum_{w \in V(T_v')\cap R}w(u) = d$.
The algorithm itself is almost the same with the previous one, but the running time analysis is slightly different.
Let $W$ be the total weight of the vertices of $T$.
The straightforward evaluation is that for each $v \in V(T)$, the values $\wbcs(v, \ast)$ are computed by a dynamic programming algorithm, which runs in $O(p_vW^2)$ time, where $p_v$ is the number of children of $v$.
Therefore, the overall running time is upper bounded by $O(nW^2)$.
To improve the quadratic dependency of $W$, we can exploit the heavy-light recursive dynamic programming technique~\cite{treedp}.
They proved that, given a tree whose vertex contains an item and each item has a weight and a value,
the problem, called {\em tree constrained knapsack problems}, of maximizing the total value of items that induces a subtree subject to the condition that the total weight is upper bounded by a given budget can be solved in $O(n^{\log3}W) = O(n^{1.585}W)$ time, where $W$ is the total weight of items.
WBCS can be seen as this tree constrained knapsack problem and then almost the same algorithm works as well.
Therefore, WBCS can be solved in $O(n^{1.585}W)$ time.

\begin{theorem}
    WBCS on trees can be solved in $\min\{O(nW^2), O(n^{1.585}W)\}$ time, where $W$ is the total weight of the vertices.
\end{theorem}

\section{Interval graphs}\label{sec:interval}

In this section, we show that BCS can be solved in $O(n^3)$ time on interval graphs.
Very recently, another polynomial time algorithm for interval graphs has been developed by \cite{bcs2}.

A graph $G = (V, E)$ is {\em interval} if it has an interval representation: an {\em interval representation} of $G$ is a set of intervals that corresponds to its vertex set $V$, such that two vertices $u, v \in V$ are adjacent to each other in $G$ if and only if the corresponding intervals have a non-empty intersection.
We denote by $I_v$ the interval corresponding to vertex $v$ and by $l_v$ and $r_v$ the left and right end points of $I_v$, respectively.
Hence, in what follows, we do not distinguish between vertices and intervals and interchangeably use them. 
Given an interval graph, we can compute an interval representation in linear time \cite{interval}.
Moreover, we can assume that, in the interval representation,
every end point of intervals has a unique integer coordinate between $1$ and $2n$.

First, we sort the input intervals in ascending order of their left end points, that is, for any $I_{v_i} = [l_i, r_i]$ and $I_{v_j} = [l_j, r_j]$ with $i < j$, it holds that $l_i < l_j$.
The following lemma is crucial for our dynamic programming.

\begin{lemma}\label{lem:interval_elimination}
    Let $S$ be a non-empty subset of $V$ such that $G[S]$ is connected and let $v$ be the vertex in $S$ whose index is maximized.
    Then $G[S \setminus \{v\}]$ is connected.
\end{lemma}
\begin{proof}
    Suppose for contradiction that $G[S \setminus \{v\}]$ has at least two connected components, say $C_a$ and $C_b$.
    An important observation is that an interval graph is connected if and only if the union of their intervals forms a single interval.
    Thus, $C_a$ and $C_b$ respectively induce intervals $\mathcal I_a$ and $\mathcal I_b$ that have no intersection with each other.
    Without loss of generality, we may assume that $\mathcal I_a$ is entirely to the left of $\mathcal I_b$, i.e., the right end of $\mathcal I_a$ is strictly to the left of the left end of $\mathcal I_b$.
    Since $G[S]$ is connected, $I_v$ must have an intersection with both $\mathcal I_a$ and $\mathcal I_b$.
    This contradicts the fact that $l_u < l_v$ for every $u \in S \setminus \{v\}$.
\end{proof}

For $0 \le i \le n$, $1 \le k \le 2n$, and $-n \le d \le n$, we say that a non-empty set $S~\subseteq~\{v_1, v_2, \ldots, v_i\}$ is {\em feasible for $(i, k, d)$} 
if $S$ induces a connected subgraph of $G$, $\max_{v\in S}r_{v} = k$, and $|S \cap B| - |S \cap R| = d$.
We denote by $\bcs(i, k, d)$ the maximum cardinality set that is feasible for $(i, k, d)$.
We also define as $\bcs(i, k, d) = -\infty$ if there is no feasible subset for $(i, k, d)$.
In particular, $\bcs(0, k, d) = -\infty$ for all $1 \le k \le 2n$ and $-n \le d \le n$.
Let $f: V \rightarrow \{1, -1\}$ such that $f(v) = 1$ if and only if $v \in B$.
The algorithm is based on the following recurrences.
\begin{lemma}\label{lem:interval}
    For $i > 0$, $\bcs(i, k, d) = $
    \begin{eqnarray*}
        \begin{cases}
            \displaystyle\max\{\bcs(i - 1, k, d - f(v_i)) + 1, \bcs(i - 1, k, d)\} & (k > r_i)\\
            \displaystyle \max_{l_i < k' < r_i}\bcs(i - 1, k', d - f(v_i)) + 1 & (k = r_i\ {\text and}\ d \neq f(v_i))\\
            \displaystyle \max\{1, \max_{l_i < k' < r_i}\bcs(i - 1, k', d - f(v_i)) + 1\} & (k = r_i\ {\text and}\ d = f(v_i))\\
            \bcs(i - 1, k, d) & (\text{otherwise}).
        \end{cases}
    \end{eqnarray*}
\end{lemma}

\begin{proof}
    We first show that the left-hand side is at most the right-hand side in all cases.
    Let $S \subseteq \{v_1, v_2, \ldots ,v_i\}$ be feasible for $(i, k, d)$ with $|S| = \bcs(i, k, d)$.
    Suppose first that $v_i \notin S$. This implies that $k \neq r_i$.
    In this case, $S$ is also feasible for $(i - 1, k, d)$ and hence we have $\bcs(i, k, d) \le \bcs(i - 1, k, d)$.
    Suppose otherwise that $v_i \in S$. By the definition of feasibility, it holds that $k \ge r_i$.
    If $S = \{v_i\}$, it holds that $\bcs(i, k, d) = 1$ if and only if $k = r_i$ and $d = f(v_i)$.
    Thus, the left-hand side is at most the right-hand side in the third recurrence.
    Let $S' = S \setminus \{v_i\}$ be non-empty. 
    Suppose moreover that $k > r_i$, that is, $\max_{v \in S'} r_v = k$.
    Since $v_i$ has the maximum index among $S$, by Lemma~\ref{lem:interval_elimination}, $G[S']$ is connected.
    Moreover, $|S' \cap R| - |S' \cap B| = |S \cap R| - |S \cap B| - f(v_i)$ holds.
    Therefore, $S'$ is feasible for $(i - 1, k, d - f(v_i))$, and then $\bcs(i, k, d) \le \bcs(i - 1, k, d - f(v_i)) + 1$ follows.
    Finally, if $k = r_i$, it holds that $l_i < \max_{v \in S'}r_v < r_i$.
    This follows from the fact that $S$ is connected and there are no intervals that share end points.
    Similar to the case where $k > r_i$, $S'$ is feasible for $(i - 1, \max_{v \in S'}r_v, d - f(v_i))$.
    Hence, $\bcs(i, k, d) \le \max_{l_i < k' < r_i}\bcs(i - 1, k', d - f(v_i)) + 1$.
    
    For the converse direction, we assume that $k \ge r_i$ since otherwise $\bcs(i - 1, k, d) = \bcs(i, k, d)$.
    Suppose first that $k > r_i$.
    If there is a feasible set for $(i - 1, k, d)$, this is also feasible for $(i, k, d)$ and $\bcs(i, k, d) \ge \bcs(i - 1, k, d)$ follows.
    Let $S'$ be feasible for $(i - 1, k, d - f(v_i))$ with $|S'| = \bcs(i - 1, k, d - f(v_i))$.
    Since the intervals are sorted in their left end and $k > r_i$, $S'$ contains an interval that entirely covers the interval $I_{v_i}$. 
    This means that $S := S' \cup \{v_i\}$ is connected and then feasible for $(i, k, d - f(v_i))$.
    Therefore, we have $\bcs(i, k, d) \ge |S'| + 1$.
    
    Suppose otherwise that $k = r_i$.
    Let $S'$ be feasible for $(i - 1, k', d - f(v_i))$ with $l_i < k' < r_i$ and let $S := S' \cup \{v_i\}$.
    As $S'$ contains an interval whose right end is strictly in between $l_i$ and $r_i$, $S$ is connected, and hence feasible for $(i, k, d)$.
    Therefore, $\bcs(i, k, d) \ge |S'| + 1$.
    Finally, suppose that $k = r_i$, $d = f(v_i)$.
    Even if there is no feasible set for $(i - 1, k', d - f(v_i))$ with $l_i < k' < r_i$, the singleton $\{v_i\}$ can be feasible and hence $\bcs(i, k, d) \ge 1$.
    
    Overall, the right-hand side is at most the left-hand side in all cases.
\end{proof}

\begin{theorem}
    Given an $n$-vertex interval graph, BCS can be solved in $O(n^3)$ time.
\end{theorem}

\begin{proof}
    For each $i > 0$, we can evaluate the recurrence in Lemma~\ref{lem:interval} in time $O(n^2)$ using dynamic programming and hence the theorem follows.
\end{proof}

As a special case of the results for interval graphs and trees, BCS on paths can be solved in linear time.
Let $v_1, v_2, \ldots, v_n$ be a path in which $v_i$ and $v_{i+1}$ are adjacent to each other for $1 \le i < n$.
First, we compute $\Left(d)$ that is the minimum integer $i$ such that $|\{v_1, v_2, \ldots, v_i\} \cap B| - |\{v_1, v_2, \ldots, v_i\} \cap R| = d$ for all $d$ with $-n \le d \le n$ and $\Pref(i) = |\{v_1, v_2, \ldots, v_i\} \cap B| - |\{v_1, v_2, \ldots, v_i\} \cap R|$ for all $1 \le i \le n$.
We can compute these values in $O(n)$ time and store them into a table.
Note that $\Left(0) = 0$ and some $\Left(d)$ is defined to be $\infty$ when there is no $i$ satisfying the above condition.
Using these values, for each $1 \le i \le n$, the maximum size of a balanced subpath whose rightmost index is $i$ can be computed by $i - \Left(\Pref(i))$.
Therefore, BCS on paths can be solved in linear time.

\begin{theorem}
    Given an $n$-vertex path, BCS can be solved in $O(n)$ time.
\end{theorem}

\section{Hardness for WBCS}\label{sec:hardness}

In this section, we discuss the hardness of the weighted counterpart of BCS, namely WBCS.
Bhore et al. \cite{bcs} proved that BCS is respectively solvable in polynomial time on trees, split graphs, and properly colored bipartite graphs.
However, we prove in this section that WBCS is hard even on those graph classes.

\begin{theorem}\label{thm:wbcs:star}
    WBCS is NP-hard even on properly colored star graphs.
\end{theorem}
\begin{proof}
    We can easily encode the subset sum problem into WBCS on star graphs.
    The subset sum problem asks for, given a set of integers $S = \{s_1, s_2, \ldots, s_n\}$ and an integer $B$, a subset $S' \subseteq S$ whose sum is exactly $B$, which is known to be NP-complete~\cite{GJ79}. 
    We take a blue vertex of weight $B$, add a red vertex of weight $s_i$ for each $s_i \in S$, and make adjacent each red vertex to the blue vertex.
    It is easy to see that the obtained graph has a feasible solution if and only if the instance of the subset sum problem has a feasible solution.
\end{proof}

Let us note that WBCS can be solved in pseudo-polynomial time on trees using the algorithm described in Section~\ref{sec:tree}.
However, WBCS is still hard on properly colored bipartite graphs and split graphs even if the total weight is polynomially upper bounded.

\begin{theorem}\label{thm:wbcs:bipartite}
    WBCS is strongly NP-hard even on properly colored bipartite graphs.
\end{theorem}
\begin{proof}
    The reduction is performed from the Exact 3-Cover problem, where given a finite set $E$ and a collection of three-element subsets $\mathcal F = \{S_1, S_2, \ldots, S_n\}$ of $E$, the goal is to find a subcollection $\mathcal F' \subseteq \mathcal F$ such that $\mathcal F'$ is mutually disjoint and entirely covers $E$.
    This problem is known to be NP-complete~\cite{GJ79}.
    
    For an instance ($E$, $\mathcal F$) of the Exact 3-Cover problem, we construct a bipartite graph $G = (V_E \cup V_{\mathcal F} \cup \{w\}, E_{\mathcal F} \cup E_w)$ as:
    $V_E = \{v_e : e \in E\}$, $V_{\mathcal F} = \{V_S : S \in \mathcal F\}$, $E_{\mathcal F} = \{\{v_e, v_S\} : e \in E, S \in \mathcal F, v_e \in S\}$, and $E_w = \{\{w, v_S\}: S \in \mathcal F\}$.
    We color the vertices of $V_{\mathcal F}$ with red and the other vertices with blue.
    Indeed, the graph obtained is bipartite and properly colored.
    We may assume that $n = 3k$ for some integer $k$.
    We assign weight one to each $v_e \in V_E$, weight $n^2$ to each $v_S \in V_{\mathcal F}$, and weight $k(n^2 - 3)$ to $w$.
    In the following, we show that $\mathcal F$ has a solution if and only if $G$ has a solution of total weight at least $2kn^2$.
    
    Let $\mathcal F' \subseteq \mathcal F$ be a solution of Exact 3-Cover.
    We choose $w$, all the vertices of $V_E$, and $v_S$ for each $S \in \mathcal F$. Clearly, the chosen vertices have total weight $2kn^2$.
    As $\mathcal F'$ covers $E$, every vertex in $V_E$ is adjacent to some $v_S$, which is chosen as our solution.
    Moreover, every vertex in $V_{\mathcal F}$ is adjacent to $w$. This implies that the chosen vertices are connected.
    Therefore, $G$ has a feasible solution of total weight at least $2kn^2$.
    
    Conversely, let $U \subseteq V_E \cup V_{\mathcal F} \cup \{w\}$ be connected in $G$ with total weight at least $2kn^2$.
    Since the total weight of the blue vertices in $G$ is $kn^2$, we can assume that the total weight of $U$ is exactly $2kn^2$.
    This means that $U$ contains exactly $k$ vertices of $V_{\mathcal F}$.
    Let $\mathcal F' \subseteq \mathcal F$ be the subsets corresponding to $U \cap V_{\mathcal F}$.
    We claim that $\mathcal F'$ is a solution of Exact 3-Cover. To see this,
    suppose that $\mathcal F'$ does not cover $e \in E$.
    Since $U$ is connected, $v_e$ has a neighbor $v_S$ in $U \cap V_{\mathcal F}$, contradicting that $e$ is not covered by $\mathcal F'$.
\end{proof}

\begin{theorem}\label{thm:wbcs:split}
    WBCS is strongly NP-hard even on split graphs.
\end{theorem}
\begin{proof}
    Recall that a graph is {\em split} if the vertex set can be partitioned into a clique and an independent set.
    The proof is almost the same with Theorem~\ref{thm:wbcs:bipartite}.
    In the proof of Theorem~\ref{thm:wbcs:bipartite}, we construct a bipartite graph $G$ that has a solution of total weight $2kn^2$ if and only if the instance of the Exact 3-Cover problem has a solution.
    We construct a split graph $G'$ from $G$ by adding an edge for each pair of vertices of $V_{\mathcal F}$.
    Analogously, we can show that $G'$ has a solution of total weight $2kn^2$ if and only if the instance of the Exact 3-Cover problem has a solution.
\end{proof}

\section{General graphs}\label{sec:general}

Since BCS is NP-hard~\cite{bcs}, efficient algorithms for general graphs are unlikely to exist.
From the viewpoint of exact exponential-time algorithms, the problem can be solved in time $O^*(1.709^n)$ using the algorithm due to Bj\"orklund et al. \cite{motif}, discussed in Section~\ref{sec:intro}.
In this section, we improve this running time to $O^*(2^{n/2})$ by modifying the well-known Dreyfus-Wagner algorithm for the minimum Steiner tree problem~\cite{DW}.

Before describing our algorithm, we briefly sketch the Dreyfus-Wagner algorithm and its improvement by Bj\"orklund et al.~\cite{subsetconv}.
The minimum Steiner tree problem asks for, given a graph $G = (V, E)$ and a terminal set $T \subseteq V$, a connected subgraph of $G$ that contains all the vertices of $T$ having the least number of edges.
The Dreyfus-Wagner algorithm solves the minimum Steiner tree problem in time $O^*(3^{|T|})$ by dynamic programming.
For $S \subseteq T$ and $v \in V$, we denote by $\opt(S, v)$ the minimum number of edges in a connected subgraph of $G$ that contains $S \cup \{v\}$.
Assume that $|S \cup \{v\}| \ge 3$ as otherwise, $\opt(S, v)$ can be computed in polynomial time.
Let $F$ be a connected subgraph that contains $S \cup \{v\}$ with $|E(F)| = \opt(S, v)$.
Note that $F$ must be a tree as otherwise we can delete at least one edge from $F$ without being disconnected.
A key observation for applying the below algorithm is that every leaf of $F$ belongs to $S \cup \{v\}$.
This enables us to decompose $F$ into (at most) three parts.
Suppose first that $v$ is a leaf of $F$.
Then, there is $w \in V(F)$ such that the edge set of $F$ can be partitioned into three edge disjoint subtrees $F_1$, $F_2$, and $F_3$: $F_1$ is a shortest path between $v$ and $w$, $F_2$ and $F_3$ induce a non-trivial partition of $S \cup \{w\}$, that is, $V(F_2) \cap (S \cup \{w\})$ and $V(F_3) \cap (S \cup \{w\})$ are both non-empty.
Suppose otherwise that $v$ is an internal vertex of $F$. We can also partition the edges of $F$ into two edge disjoint subtrees $F_1$ and $F_2$ such that $F_1$ contains $S' \cup \{w\}$ and $F_2$ contains $(S \cup \{w\}) \setminus S'$ for some non-empty proper subset $S'$ of $S$.
This leads to the following recurrence.
\[
    \opt(S, v) = \min_{w \in V}\left\{d(v, w) + \min_{\substack{S' \subset S\\S' \neq \emptyset}} \left(\opt(S', w) + \opt(S \setminus S', w)\right)\right\}.
\]
Note that if $v$ is an internal vertex, the minimum is attained when $v = w$ in the above recurrence.
A naive evaluation of this recurrence takes $O^*(3^{|T|})$ time in total.
Bj\"orklund et al. \cite{subsetconv} proposed a fast evaluation technique for the above recurrence known as the {\em fast subset convolution}, described in Theorem~\ref{thm:fsc}, which allows us to compute $\opt(S, v)$ for all $S \subseteq T$ and $v \in V$ in total time $O^*(2^{|T|})$.

\begin{theorem}[\cite{subsetconv}]\label{thm:fsc}
    Let $U$ be a finite set.
    Let $M$ be a positive integer and let $f, g: 2^U \rightarrow \{0, 1, \ldots, M, \infty\}$.
    Then, the subset convolution over the min-sum semiring
    \[
        (f\ast g)(X) = \min_{Y \subseteq X} (f(Y) + g(X \setminus Y))
    \]
    can be computed in $2^{|U|}(|U| + M)^{O(1)}$ time in total for all $X \subseteq U$.
\end{theorem}

For our problem, namely BCS, we first solve a variant of the minimum Steiner tree problem defined as follows.
Let $G = (V, E)$ be the instance of BCS.
Without loss of generality, we assume that $|R| \le |B|$.
For $S \subseteq R$ and $v \in V \setminus S$, we compute $\opt'(S, v)$ the minimum number of edges of a tree connecting all the vertices of $S \cup \{v\}$ and {\em excluding any vertex of $R \setminus (S \cup \{v\})$}.
The recurrence for $\opt'$ is quite similar to one for the ordinary minimum Steiner tree problem, but an essential difference from the above recurrence is that the shortest path between $v$ and $w$ does not contain any red vertex other than $S \cup \{v\}$. 

\begin{lemma}
    For $S \subseteq R$ and for $v \in V$,
    \[
    \opt'(S, v) = \min_{w \in V \setminus (R \setminus (S \cup \{v\}))}\left\{d'(v, w) + \min_{\substack{S' \subset S\\S' \neq \emptyset}} \left(\opt'(S', w) + \opt'(S \setminus S', w)\right)\right\},
    \]
    where $d'(v, w)$ is the number of edges in a shortest path between $v$ and $w$ excluding red vertices except for its end vertices.
    If there is no such a path, $d'(v, w)$ is defined to be $\infty$.
\end{lemma}

\begin{proof}
    The idea of the proof is analogous to the Dreyfus-Wagner algorithm for the ordinary Steiner tree problem. 
    Suppose that $F$ is an optimal solution that contains every vertex in $S \cup \{v\}$ and does not contain every vertex in $R \setminus (S \cup \{v\})$.
    Similarly, we can assume that every leaf of $F$ belongs to $S \cup \{v\}$ as otherwise such a leaf can be deleted without losing feasibility.
    Then, the edge set of $F$ can be partitioned into three edge disjoint subtrees $F_1$, $F_2$, and $F_3$ such that $F_1$ is a shortest path between $v$ and $w$, $F_2$ and $F_3$ induces a non-trivial bipartition of $S$.
    The only difference from the ordinary Steiner tree problem is that $F_1$, $F_2$ and $F_3$ should avoid any irrelevant red vertex.
    This can be done since $F_1$ does avoid such a vertex.
\end{proof}

Similar to the normal Steiner tree problem, we can improve the naive running time $O^*(3^{|R|})$ to $O^*(2^{|R|})$ by the subset convolution algorithm in Theorem~\ref{thm:fsc}.

Now, we are ready to describe the final part of our algorithm for BCS. We have already know $\opt'(S, v)$ for all $S \subseteq V$ and $v \in V \setminus (R \setminus S)$.
Suppose $\opt'(S, v) < \infty$. Let $F$ be a tree with $|E(F)| = \opt'(S, v)$ such that $V(F) \cap R = (S \cup \{v\}) \cap R$.
Such a tree can be obtained in polynomial time using the standard traceback technique.
Since $F$ is a tree, we know that $|V(F)| = \opt'(S, v) + 1$.
Let $k$ be the number of red vertices in $F$.
If $F$ contains more than $k$ blue vertices, we can immediately conclude that there is no balanced connected subgraph $H$ with $V(H) \cap R = V(F) \cap R$.
Suppose otherwise.
The following lemma ensures that an optimal solution of BCS can be computed from some Steiner tree.

\begin{lemma}
    Let $R_F = V(F) \cap R$. If there is a balanced connected subgraph $H$ with $V(H) \cap R = R_F$, then there is a balanced connected subgraph $H'$ with $V(S)' \cap R = R_F$ such that $F$ is a subtree of $H'$. Moreover, such a subgraph $H'$ can be constructed in linear time from $F$.
\end{lemma}

\begin{proof}
    We prove this lemma by giving a linear time algorithm that constructs a balanced connected subgraph $H'$ when given $F$ as in the statement of the lemma.
    Suppose that there is a balanced connected subgraph $H$ with $V(H) \cap R = R_F$. Since $F$ is a minimum Steiner tree with $V(F) \cap R = R_F$, the number of blue vertices in $F$ is not larger than that of red vertices.
    We greedily add a blue vertex that has a neighbor in $F$ to $F$ as long as it is not balanced.
    We claim that it is possible to construct a balanced connected subgraph using this procedure.
    Let $H'$ be a maximal graph that is obtained by the above procedure. Suppose for contradiction that $|V(H')| < |V(H)|$.
    Let $v \in (V(H) \setminus V(H')) \cap B$ be a blue vertex and let $r$ be a red vertex in $F$.
    We choose $v$ and $r$ in such a way that the distance between $v$ and $r$ in $H$ is as small as possible.
    Consider a shortest path between $v$ and $r$ in $H$. By the choice of $v$ and $r$,
    there are no red vertices other than $r$ and no blue vertices of $V(H) \setminus V(H')$ on the path.
    Since $v \notin V(H')$ and $r \in V(H)$, there is a vertex $v'$ on the path such that $v'$ has a neighbor in $V(H')$ but not contained in $H'$, which contradicts the maximality of $V(H')$. 
    
    This greedy algorithm runs in linear time and hence the lemma follows.
\end{proof}

Overall, we have the following running time.

\begin{theorem}
    BCS can be solved in $O^*(2^{n/2})$ time.
\end{theorem}

\section*{Acknowledgements}
This work is partially supported by JSPS KAKENHI Grant Number JP17H01788 and JST CREST JPMJCR1401.


\begin{thebibliography}{10}

\bibitem{motif}
Bj\"orklund, A., Kaski, P., Kowalik, \L.: Constrained Multilinear Detection and Generalized Graph Motifs.
Algorithmica 74(2), 947–967 (2016)

\bibitem{subsetconv}
Bj\"orklund, A., Husfeldt, T., Kaski, P., Koivisto, M.:
Fourier meets m\"obius: fast subset convolution.
In Proceedings of STOC 2007, pp.~67--74 (2007)

\bibitem{bcs}
Bhore, S., Chakraborty, S., Jana, S., Mitchell, J. S., Pandit, S., Roy, S.: The Balanced Connected Subgraph Problem.
In Proceedings of CALDAM 2019, LNCS vol.~11394, pp.~201--215 (2019)

\bibitem{bcs2}
Bhore, S., Jana, S., Mitchell, J. S., Pandit, S., Roy, S.: Balanced Connected Subgraph Problem in Geometric Intersection Graphs.
ArXiv:1909.03872 (2019)

\bibitem{tddp}
Bodlaender, H. L., Cygan, M., Kratsch, S., Nederlof, J.:
Deterministic single exponential time algorithms for connectivity problems parameterized by treewidth.
Information and Computation 243, 86--111 (2015)

\bibitem{interval}
Booth, K. S., Lueker, G. S.: Testing for the consecutive ones property, interval graphs, and graph planarity using PQ-tree algorithms.
Journal of Computer and System Sciences 13(3), 335--379 (1976)

\bibitem{fregression}
Calders, T., Karim, A., Kamiran, F., Ali, W., Zhang, X.:
Controlling Attribute Effect in Linear Regression.
In Proceedings of ICDM2013, IEEE, pp.~71--80 (2013)

\bibitem{franking}
Celis, L. E., Straszak, D., Vishnoi, N. K.:
Ranking with Fairness Constraints.
In Proceedings of ICALP 2018, LIPIcs, pp.~28:1--28:15 (2018)

\bibitem{fcluster}
Chierichetti, F., Kumar, R., Lattanzi, S., Vassilvitskii, S.: Fair clustering through fairlets.
In Proceedings of NIPS 2017, pp.~5029–-5037 (2017)

\bibitem{fairmatching}
Chierichetti, F., Kumar, R., Lattanzi, S., Vassilvitskii, S.:
Matroids, Matchings, and Fairness.
In Proceedings of AISTATS 2019, pp.~2212--2220 (2019)

\bibitem{mso1}
Courcelle, B.: The monadic second-order logic of graphs. I. Recognizable sets of finite graphs.
Information and computation 85(1), 12–75 (1990)

\bibitem{mso2}
Courcelle, B., Mosbah, M.: Monadic second-order evaluations on tree-decomposable graphs.
Theoretical Computer Science 109(1), 49–82 (1993)

\bibitem{param}
Cygan, M., Fomin, F. V., Kowalik, \L., Lokshtanov, D., Marx, D., Pilipczuk, M., Pilipczuk, M., Saurabh, S.:
Parameterized Algorithms. Springer (2015)

\bibitem{DW}
Dreyfus, S. E., Wagner, R. A.:
The Steiner problems. Network 1(3), 195--207 (1971)

\bibitem{GJ79}
Garey, M. R., Johnson, D. S.:
Computers and Intractability: A Guide to the Theory of Np-Completeness.
W H Freeman \& Co (1979)

\bibitem{fbandit}
Joseph, M., Kearns, M., Morgenstern, J. H., Roth, A.:
Fairness in Learning: Classic and Contextual Bandits.
In Proceedings of NIPS 2016, pp.~325--333 (2016)

\bibitem{treedp}
Kumabe, S., Maehara, T., Sin'ya, R.:
Linear Pseudo-Polynomial Factor Algorithm for Automaton Constrained Tree Knapsack Problem.
In Proceedings of WALCOM 2019, LNCS, vol.~11355, pp.~248--260 (2019)

\end{thebibliography}
\end{document}